\documentclass[10pt]{article}

\usepackage{latexsym}
\usepackage{algorithm}
\usepackage{algorithmic}
\usepackage{graphicx}
\usepackage{subfigure}
\usepackage[T1]{fontenc}
\usepackage{mathptmx}
\usepackage{charter}
\usepackage{amsmath}

\newtheorem{theorem}{Theorem}
\newtheorem{lemma}{Lemma}

\newtheorem{definition}{Definition}

\newtheorem{remark}{Remark}

\newenvironment{proof}
        {\noindent {\em Proof.}~~~} 
        {\begin{flushright}$\Box$\end{flushright}}

\title{Sharp utilization thresholds for some real-time scheduling problems}

\author{
	Sathish Gopalakrishnan \\
	Department of Electrical and Computer Engineering \\
	The University of British Columbia
}

\date{}

\begin{document}

\maketitle

\setlength{\baselineskip}{1\baselineskip}


\begin{abstract}

Scheduling policies for real-time systems exhibit threshold behavior that is related to the utilization of the task set they schedule, and in some cases this threshold is sharp. For the rate monotonic scheduling policy, we show that periodic workload with utilization less than a threshold $U_{RM}^{*}$ can be scheduled almost surely and that all workload with utilization greater than $U_{RM}^{*}$ is almost surely not schedulable. We study such sharp threshold behavior in the context of processor scheduling using static task priorities, not only for periodic real-time tasks but for aperiodic real-time tasks as well. The notion of a utilization threshold provides a simple schedulability test for most real-time applications. These results improve our understanding of scheduling policies and provide an interesting characterization of the typical behavior of policies. The threshold is sharp (small deviations around the threshold cause schedulability, as a property, to appear or disappear) for most policies; this is a happy consequence that can be used to address the limitations of existing utilization-based tests for schedulability. We demonstrate the use of such an approach for balancing power consumption with the need to meet deadlines in web servers.

\end{abstract}


\section{Introduction}

Computing systems have become larger in scale and more pervasive in their applications. The constant interaction between embedded computing systems and the physical world requires a notion of predictable behavior from the deployed computing systems. Even in large-scale computing clusters and server farms there is a growing emphasis on providing service guarantees. This need for predictable operation can often be characterized by a need for timely completion of activities. Tasks can usually be associated with deadlines; systems need to ensure that the tasks meet their deadlines.

In a sense, the convergence of computation, communication and control, which is often seen in distributed embedded systems, has led to a renewed interest in understanding the conditions for a system to meet deadlines. Additionally, most tasks are recurring: they need to be performed repeatedly because of the constant interaction with the physical environment (or because of user demand). Such problems have been at the heart of real-time scheduling since the seminal work by Liu and Layland~\cite{LL:73} on utilization bounds for schedulability using static and dynamic priority scheduling policies.

Liu and Layland considered a set of periodic tasks with known execution times and periods that need to be scheduled on a uniprocessor system. Each task $\tau_i$ was characterized by its execution time $c_i$ and its period $P_i$. In the periodic task model, if an instance of task $\tau_i$ is eligible for execution at time $t$, the next instance of the same task is eligible for execution at time $t+P_i$. Each instance of a task is called a {\em job}. Liu and Layland restricted their analysis to task sets where each job needs to complete before the next job belonging to the same task is ready for execution. For task $\tau_i$, each job needs to complete execution within $P_i$ time units after its release. Hence $P_i$ is known as the {\em relative deadline} for the task. It is easy to see that each task will use the processor for a $\frac{c_i}{P_i}$ fraction of time. This fraction is the utilization of task $\tau_i$ and can be denoted $u_i$. The utilization of a task set, therefore, is $U = \sum_{i=1}^{n} u_i$ where $n$ is the number of tasks.

The fundamental contribution that Liu and Layland made was to show that for a specific scheduling policy $\zeta$ -- they studied the Rate Monotone policy and the Earliest Deadline First policy -- there exists a utilization bound $U_\zeta$ such that any task set with utilization $U < U_\zeta$ is definitely schedulable (all deadlines will be met). This has formed the basis for much work in real-time systems.

There are, however, some obvious limitations to Liu and Layland's result. The first drawback is that the utilization bound test is pessimistic: there are many task sets that may exceed the bound but
are still schedulable. Second, for models when the relative deadline does not equal the period, additional tests are needed. Lastly, obtaining the utilization bound is difficult for many policies
because such derivations involve identifying the worst-case task set (the task set with low utilization that is not schedulable) and this is non-trivial for certain policies. 

In contrast with prior work on schedulability and predictability, we show that the rate monotonic scheduling policy has a utilization threshold $U_{RM}^*$ such that any task set with utilization less than $U_{RM}^*$ is almost surely schedulable and a task set with utilization greater than $U_{RM}^*$ is almost surely not schedulable. Similarly, we show that such a threshold exists for deadline monotonic scheduling of aperiodic real-time tasks. Establishing the sharpness of utilization thresholds provides a better understanding of scheduling policies and removes most of the pessimism that is associated with traditional utilization bounds because of the implication that task sets with utilization greater than $U^*$ are unlikely to be schedulable. These results are independent of the relationship between task periods and task deadlines. On the other hand, it is prudent to note that these results indicate that schedulability appears and disappears {\em almost surely}. For hard real-time systems, which cannot afford to miss any deadlines, this suggests that the threshold can be used as an initial estimate and schedulability needs to be verified by an exact test at some step. For soft real-time systems, which can tolerate some deadline misses, our results provide a simple test and a tight performance guarantee.

As an example, consider rate monotonic scheduling with the Liu and Layland task model. We would like to show that when $n$, the number of tasks to schedule, is large, almost surely task sets of utilization less than about $0.80$ utilization are schedulable and almost surely task sets with greater utilization are unschedulable. This shows that the average performance of the rate monotonic policy is much better than the Liu and Layland worst-case utilization of $0.69$~\cite{LL:73}. It is exactly for the case of large task sets that other analysis techniques become computationally expensive. The fundamental contribution we make is a framework for answering questions about average or typical case schedulability. To date there has been no unified methodology that can deal with all scheduling policies.

In this article, our emphasis is on rate monotonic scheduling for periodic tasks and deadline monotonic scheduling for aperiodic tasks on a uniprocessor although some preliminary experiments lead us to believe that these results should hold for multiprocessor and distributed (multistage) systems as well.

{\bf Motivation.} The main reason for studying sharp thresholds is to ease resource provisioning for soft real-time systems, and, in some cases, simplify the offline optimization of hard real-time systems. The existence of sharp thresholds allows us to make efficient use of computing resources. Many mainstream operating systems (especially Linux) support simple fixed-priority scheduling and being able to identify a workload limit for such systems allows for simple admission control and resource management. Many applications have tasks with deadlines but are built to tolerate a few deadline misses. Multimedia applications have been traditional examples, but many emerging pervasive computing applications are of a similar nature. Timely response leads to high quality of service but occasional delays are not catastrophic. For these systems, being able to utilize resources better can lead to substantial cost savings that will allow these applications to achieve greater market penetration. It can be argued that feedback control~\cite{LSTS:02} can keep soft real-time systems in an acceptable operation regime but such techniques may require substantial modifications to operating systems and/or middleware platforms. Additionally, our findings may allow feedback control schemes to pick better set points.

In the next section (Section~\ref{sec:model}), we elaborate on the model for periodic real-time tasks that we consider and the notation we will follow. Then, in Sections~\ref{sec:thresholds}, \ref{sec:graph-property} and \ref{sec:sharpness}, we will develop the framework for reasoning about average-case behavior and present proofs of our key results. We will then followup with experimental evidence and discussion of the results (Section~\ref{sec:expts}). We then extend these results to the aperiodic task model and demonstrate the use of sharp threshold behavior in power control for web servers (Section~\ref{sec:aperiodic}). Finally, we place our work in context with related work (Section~\ref{sec:relatedwork}) and conclude the article (Section~\ref{sec:conclusions}).

\section{System and task models}
\label{sec:model}

We consider a general and well-understood model for uniprocessor scheduling. 

{\bf Platform model.} We consider a uniprocessor system that can schedule tasks using static priorities and preempt (suspend execution of) tasks to schedule tasks with higher priority.

{\bf Task model.}
Each task $\tau_i$ is periodic with period $P_{i}$. Each instance of the task has an execution time requirement $c_{i}$ on the processor and a {\em relative deadline} $D_{i} = P_i$. If a job of $\tau_i$ is released (ready for execution) at at time $t$ then it is expected to finish execution by time $t+P_i$. Tasks are independent of each other.

The typical assumption is that the first instance of all periodic tasks release at the same instant in time. A reason for making this assumption is that this represents the worst-case situation for static priority policies. We will also make this assumption although it is not strictly necessary.

The utilization of a periodic task set is \[ U := \sum_{i}\frac{c_{i}}{P_{i}}. \]

{\bf Monotone scheduling policies.} In this article, we will mostly be concerned with the rate monotonic and deadline monotonic scheduling policies, which are work-conserving (non-idling) policies. It is also useful to keep in mind a more general classification of policies: the class of monotone policies. Let us suppose that a scheduling policy successfully schedules a set of tasks $\Gamma=\{\tau_{i}\}$. We will call the policy a {\em monotone scheduling policy}\footnote{Note that there is a distinction between this notion of monotonicity and the use of the term ``monotone'' in the context of rate/deadline monotone priority policy. However, by this definition, the rate monotonic scheduling policy and the deadline monotonic scheduling policy are monotone scheduling policies.} if and only if:
\begin{itemize}
	\item It can schedule any set $\Delta \subset \Gamma$ successfully;
	\item For any task $\tau_{i} \in \Gamma$, the policy can schedule all tasks successfully if $c_{i}$ were to be reduced;
	\item For any task $\tau_{i} \in \Gamma$, the policy can schedule all jobs successfully if $P_{i}$ were increased.
\end{itemize}

\section{Utilization thresholds}
\label{sec:thresholds}

Let $\Gamma$ be some task set. We define $\nu(U,\Gamma)$ as the probability of selecting task set $\Gamma$ from all possible task sets of utilization $U$. Let $S_n$ represent the set of all schedulable task sets with $n$ tasks. Then $\mu(U,S_n) := \sum_{\Gamma \in S_{n}} \nu(U,\Gamma)$ represents the probability that a task set with utilization $U_{n}$ is schedulable using the rate monotonic policy. This can also be stated in the following manner. Suppose $\Gamma$, a task set with $n$ tasks, is drawn at uniformly at random from the space of all possible task sets of utilization $U_{n}$. Then, $\mu(U,S_n)$ is the probability of the event ``$\Gamma \in S_n$.''

\begin{definition}[Threshold] $U^{*}_n$ is said to be a threshold for $S_n$ if for any $U$
\begin{equation}
	\lim_{n \rightarrow \infty} \mu(U,S_n) = \left\{ \begin{array}{ll}
												0 & \textrm{if } U \gg U^{*}_{n},\\
												1 & \textrm{if } U \ll U^{*}_{n}.
												\end{array}
											   \right.
\end{equation}
\end{definition}
Note that $f \ll g$ means $f/g \rightarrow 0$.

The definition of threshold may appear trivial in the case of scheduling policies (clearly utilization of $0$ is schedulable, and utilization $> 1$ is unschedulable) and hence we require a stronger criterion for a useful threshold.

\begin{definition}[Sharp threshold] A threshold is said to be sharp if there exists a $U^{*}_n$ such that for every $\epsilon > 0$ and any $U$
\begin{equation}
	\lim_{n \rightarrow \infty} \mu(U,S_n) = \left\{ \begin{array}{ll}
												0 & \textrm{if } U > (1+\epsilon)U^{*}_{n},\\
												1 & \textrm{if } U < (1-\epsilon)U^{*}_{n}.
												\end{array}
											   \right.
\end{equation}
\end{definition}

The interval of width $2\epsilon$ over which the probability of finding a valid schedule drops from $1$ to $0$ is called the {\em threshold interval}. As $n \rightarrow \infty$, the threshold interval becomes arbitrarily small and we have a sharp threshold.

A threshold that is not sharp is a coarse threshold. Sharp thresholds represent phase transition phenomena because we can divide the task set space into two phases: one in which the property holds almost always and one in which it almost always does not hold.

We emphasize once more that, although the results are asymptotic, in practice a reasonable number of tasks suffices for observing sharp thresholds. When we think of $n \rightarrow \infty$, we do not conjure up task sets with $1000$s of tasks; we are usually dealing with many $10$s of tasks.

{\em The main result of our work is that schedulability, with the rate monotonic scheduling policy, of periodic tasks has a sharp threshold.}

By proving such a result we provide a platform for the average-case analysis of real-time scheduling problems and highlight the validity of using empirical utilization thresholds for managing resource allocation.

\section{Schedulability as a graph property}
\label{sec:graph-property}

To show that scheduling problems of the type that we are interested in have a sharp threshold we will gain leverage from work carried out in the context of random graphs. The study of phase transitions can be traced back to the work of Erd\"os and R\'enyi on random graphs~\cite{ER:59,ER:60}. A random graph is a graph with a fixed set of vertices and edges between two given nodes occur with some probability, $p$. Erd\"os and R\'enyi showed that as the parameter controlling the edge probability varies, the random graph system experiences a swift qualitative change. This transition is similar to observations in the physical world. Akin to water freezing abruptly as its temperature drops below zero, the random graph changes rapidly from having many small components to a graph with a giant component that contains a constant proportion of vertices.

We will use results that have been obtained by Friedgut and Bourgain~\cite{Friedgut:99} to prove the existence of a sharp utilization threshold for schedulability. The first step is, of course, to connect the scheduling problem to a problem on graphs.

To consider scheduling as a graph problem, we will first deal with utilization in a quantized fashion. Let $q$ be the smallest quantum of utilization that can be allocated to a task. Each task can have a utilization of at most $1$ therefore there are at most $M = 1/q$ quanta, and $M$ is assumed to be sufficiently large. More specifically, given $n$ tasks, without loss of generality, $M$ can be of the form $(k_{2})^{n}$ for some constant $k_{2} > 2$. If $u_{i}$ is the utilization of task $\tau_{i}$ and $P_{i}$ is the period of the task, then $c_{i} = u_{i}P_{i}$. We can thus represent each periodic task by the tuple $\{P_{i},u_{i}\}$.\footnote{The use of quantized utilization does not limit our analysis in any way; $M$ can be made sufficiently large to approximate real allocations.}

Consider a bipartite graph with the two vertex partitions being $\mathcal{T}$ and $\mathcal{U}$. The vertices in $\mathcal{T}$ represent tasks and each vertex can be labeled by its period. (The periods, as can be expected, are assumed to be chosen uniformly at random from the space of all possible periods.) The set $\mathcal{U}$ contains $M$ vertices, each corresponding to one quantum of utilization. The complete bipartite graph with $\mathcal{T}$ and $\mathcal{U}$ as the two partitions represents a task set with each task having utilization $1$. This is clearly unschedulable for task sets with more than one task. If edges are present with probability $p$ then we have random task sets with an expected utilization of $Mnp$ where $n$ is the number of tasks. By choosing the value of $p$ appropriately, we can generate random task sets with varying utilization levels. There is a graph corresponding to each task set and we will call these graphs {\em task set graphs}. In turn, for each utilization level, there is a corresponding edge probability $p$ (The complete bipartite graph representation is illustrated in Figure~\ref{fig:sch-graph}.). The set of periods is a set of integers. For $n$ tasks, there may be at most $n$ periods chosen from a range of integers. When $n$ is large, we can represent all possible periods using such a graph.

\begin{figure}[htbp]
	\centering
	\scalebox{0.65}{
	\includegraphics{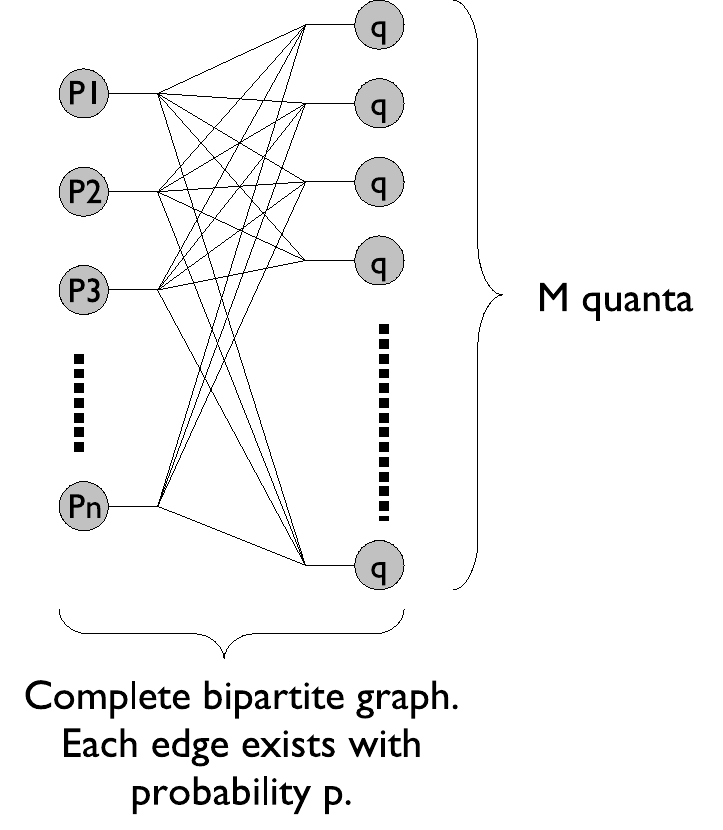}
	}
\caption{Task set graph}
\label{fig:sch-graph}
\end{figure}

Certain combinations of periods and execution times lead to unschedulable task sets under the rate monotonic scheduling policy (Figure~\ref{fig:unsch-graph} depicts a task set of utilization $0.97$ that cannot be scheduled using the rate monotonic policy. This task set has two tasks: one with period $10$ and utilization $0.8$ and another with period $18$ and utilization $0.17$.). This phenomenon is well understood from the initial study by Liu and Layland~\cite{LL:73}. These unschedulable task sets are subgraphs of the complete bipartite task set graph. Increasing $p$ from $0$ to $1$ leads to unschedulable task sets. There is, in fact, a critical edge probability, $p^{*}$, which in turn corresponds to a critical utilization $U^{*}$ for (un)schedulability. For $p < p^{*}$, the expected task set is asymptotically almost surely schedulable; for $p > p^{*}$, the expected task set is asymptotically almost surely not schedulable. The next section details the proof of this sharp threshold behavior.

{\em Remark.} In our description of the graph model, we assumed that edges in the task set graph exist with probability $p$, which would imply that only the average utilization is fixed. We can invert the model by fixing the number of edges. This does not alter the probability of an edge existing between two vertices but ensures that the utilization (and not just the average utilization) is fixed.

\begin{figure}
	\centering
	\scalebox{0.65}{
	\includegraphics{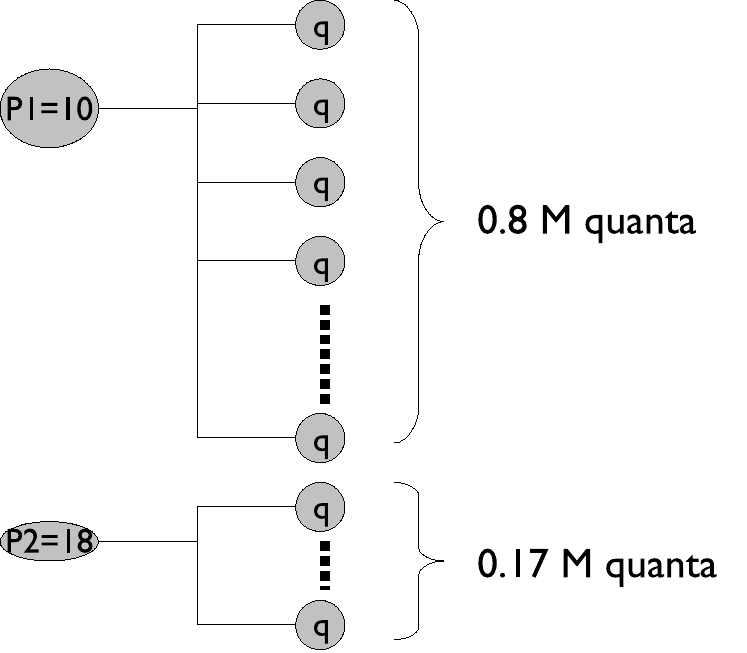}
	}
\caption{Unschedulable task set graph}
\label{fig:unsch-graph}
\end{figure}

\section{Sharpness of utilization thresholds}
\label{sec:sharpness}

The scheduling graph provides a structure to study the typical (or average) case behavior of scheduling problems. Given a utilization level, each edge in the graph appears with a certain probability that captures typical scheduling problems. 

Scheduling problems do exhibit threshold behavior and that this behavior is controlled by the utilization of the set of tasks to be scheduled. To be convinced, it should suffice to remark that when utilization is $0$ all task sets are schedulable, and for utilizations above $1$ no task set is schedulable. The primary question then becomes: ``Is the threshold sharp or is it coarse?'' To answer this question we will need further results.

\subsection{Preliminaries}
We will apply a theorem obtained by Bourgain that appeared as an appendix to Friedgut's article~\cite{Friedgut:99}. Recall that $\{0,1\}^{N}$ is the set of all $N$-bit vectors and that any $x \in \{0,1\}^{N}$ is an $N$-bit vector. We can use these vectors to indicate the presence of edges in a graph with at most $N$ edges. The size of such a vector $x$, denoted $|x|$, is the number of 1s it contains. 

Let $V = \{0,1\}^{N}$ and let $W$ be some subset of $V$ that represents a graph property. In our discussion, it will be useful to consider $V$ to be the collection of all possible task set graphs and $W$ to be the collection of unschedulable task set graphs. When the vertex partitions $\mathcal{T}$ and $\mathcal{U}$ are known, then the task set graph is defined by its edges. If $N$ is the maximum possible number of edges, then every element of $V$, i.e., an $N$-bit vector, represents a task set graph. 

The general definition of a monotone property follows, where $x$ and $y$ are elements of $V$ and can also be treated as vectors; $x_i$ is the $i^{th}$ element of the vector $x$.

\begin{definition}[Monotone property] 
$W \subset V$ is said to be monotone if and only if $\forall x \in W, y \in  V: (\forall i, y_i \ge x_i) \Rightarrow y \in W.$
\end{definition}

In the context of graphs, a monotone property is one that cannot be destroyed by the addition of edges. 

If a  task set is unschedulable using the rate monotonic policy, then increasing the utilization of any of the tasks, which adds edges to the corresponding task set graph, will also result in an unschedulable task set. The monotone graph property of interest to us is that collection of edges that makes a task set unschedulable. (Adding edges to a graph representing an unschedulable task set will result in another unschedulable task set graph.)

In the theorem that follows, the term $\mu_p(A)$ is the probability that a graph property $A$ is present if each of the $N$ edges is chosen with probability $p$. If $x$ represents a graph, and if $x_{i}$ is $1$ when edge $i$ is present in $x$ and is $0$ otherwise, then
\[ \mu_{p}(A) := \sum_{x \in A} p^{|x|}(1-p)^{N-|x|}. \]

\begin{theorem}[Bourgain~\cite{Friedgut:99}] 
\label{th:bourgain}
Let $A \subset \{0,1\}^{N}$ be a monotone property and assume say 
\begin{eqnarray}
	(A1): & \mu_p(A) = \frac{1}{2}, & \\
	(A2): & p\frac{d\mu_p(A)}{dp} < C, & \\
	(A3): & p = o(1). &
\end{eqnarray}
Then there is a $\delta = f(C)$ for some function $f$ such that at least one of the following two conditions must be true:
\begin{enumerate}
	\item[C1:] \begin{equation}
	\label{eq:ineq21}
	\mu_p(x \in \{0,1\}^{N} | x \textrm{ contains } x' \in A, |x'| \le 10C) > \delta
\end{equation}
	\item[C2:] There exists $x' \notin A$ \textrm{ of size } $|x'| \le 10C$ such that the conditional probability
\begin{equation}
	\label{eq:ineq22}
	\mu_p(x \in A | x \supset x') > \frac{1}{2}+\delta.
\end{equation}
\end{enumerate}
\end{theorem}

$f(n) \in o(g(n))$ is equivalent to stating that $\lim_{n \rightarrow \infty} \frac{f(n)}{g(n)} =  0$. The edge probabilities are functions of $N$, the maximum size of the graph, and are expected to diminish as $N$ increases. This is captured as $p=o(1)$, to indicate that $p \ll 1$.

Some comments about Bourgain's theorem are now in order. Bourgain's theorem, in essence, states that if a monotone property is such that $p\frac{d\mu}{dp} < C$ then that monotone property is approximated by a ``local property.''\footnote{Friedgut proved a similar result except that Friedgut's approach required that the random structure under investigation exhibit some symmetry~\cite{Friedgut:99}.} In a graph, a local property is a property that depends on a small number of vertices and edges. Bourgain proved that if $p\frac{d\mu_p(A)}{dp}$ is bounded by some constant, then there must exist some small graphs (whose sizes are bounded by a constant) that are capable of boosting the probability of the desired property appearing. $x'$ is such a booster. Inequality \eqref{eq:ineq21} suggests that most graphs that possess the monotone property in fact contain a subgraph that satisfies the property. Inequality \eqref{eq:ineq22} is equivalent to saying that for some graph $y \subset \{0,1\}^{N}$, the probability that $y \cup x'$ is in $A$ is at least $1/2 + \delta$.

We shall explain this result using the Erd\"os and R\'enyi model for random graphs. In this model, each of the possible $m \choose 2$ edges in a graph with $m$ vertices is added with probability $p$. The property that the random graph is connected is not a local property because it involves all the vertices in the graph. On the other hand, the property that the random graph contains a triangle is local because a triangle has only $3$ vertices. It is for this reason that connectivity has a sharp threshold~\cite{ER:59} but the existence of a triangle has a coarse threshold~\cite{Friedgut:99}. A sharp threshold is associated with a rapid change in the appearance or disappearance of a property, which means that $\frac{d\mu_p(A)}{dp} \rightarrow \infty$ when $\mu_p(A) = \frac{1}{2}$. When a threshold is coarse, this derivative (or slope) is finite.
A vital point to note is that $p\frac{d\mu_p(A)}{dp} < C$ is definitely true for all coarse thresholds and may be true for some sharp thresholds. On the other hand, $p\frac{d\mu_p(A)}{dp} \rightarrow \infty$ holds only for sharp thresholds and is a stronger characterization of certain sharp thresholds. For schedulability, we will show that this stronger result holds. To do so, we will show that schedulability depends on a non-local property of the schedulability graph. We also add that $\mu_{p}(A)$ is continuous and its derivative exists: every 

\begin{lemma}
For a set of $n$ periodic tasks with periods $P_{1},\dots,P_{n}$, the minimum utilization for a task set that is barely schedulable using the rate monotonic policy is achieved only when all $n$ tasks have specific execution times. As a result, any task set that is unschedulable will have at least $n$ edges.
\label{lem:n-edges}
\end{lemma}
\begin{proof}
From Liu and Layland's proof~\cite{LL:73}, the task set with minimum utilization that is {\em barely schedulable}\footnote{A barely schedulable task set is one that fully utilizes the processor for an interval of time that begins with the arrival of an instance of some task and extends at least up to the deadline of that task instance.} is such that $c_{i} = P_{i+1}-P_{i}, \forall i \in \{1,n-1\}$, $c_{n} = P_{n} - 2(c_{1}+\dots+c_{n-1})$, and $P_{n} > P_{n-1} > \dots > P_{1}$. Because utilization of the barely schedulable task set is minimized only when all tasks have non-zero execution times, the task set graph has at least $n$ edges for the barely schedulable task set. A task set that is unschedulable has to have a higher utilization and hence the corresponding task set graph will also have at least $n$ edges.
\end{proof} 

\subsection{Main result}

\begin{theorem}
	The schedulability of a task set with $n$ periodic tasks, where each task $\tau_i$ is characterized by execution time $c_i$, period $P_i$, and relative deadline equal to its period, has a sharp utilization threshold. The utilization of the set of tasks is $U := \sum_{i=1}^{n} \frac{c_i}{P_i}$.
	\label{th:main}
\end{theorem}
\begin{proof}
	Consider the task set graph that represents the (un)schedulability of a task set with $n$ tasks when each edge occurs with probability $p$ and the corresponding utilization level is $Mnp$. A task set is unschedulable if and only if the corresponding task set graph includes an assignment of utilizations to periods that causes deadline misses. We need to show that there is some $p^{*}$, and hence some $U^{*} = Mnp^{*}$, that is a sharp threshold.
	
	Let $A$ represent the property of a task set graph containing an unschedulable assignment of utilizations to periods. Choose a $p$ such that $\mu_p(A) = \frac{1}{2}.$ We can always find such a $p$ because we know that all task sets are schedulable when utilization is $0$ and no task set is schedulable if its utilization exceeds $1$.
	
	$\mu_{p(A)}$ is a polynomial in $p$ and is differentiable with respect to $p$. Let us suppose that $p\frac{d\mu_p(A)}{dp} < K/10$. It is also the case that $p < 1/Mn$ and hence $p=o(1)$.
	
	We will assume that the set of possible edges is $E$. From Bourgain's theorem (Theorem~\ref{th:bourgain}) we know that if all the conditions are true (especially the constraint on $p\frac{d\mu_p(A)}{dp}$) then there must exist some $x'$ such that 
	\begin{equation}
		\label{eq:ineq1}
		\mu_p(x \in \{0,1\}^{|E|})| x \textrm{ contains } x' \in A, |x'| \le K) > \delta 
	\end{equation} 
	or there exists $x' \notin A$ of size $|x'| \le K$ such that the conditional probability 
	\begin{equation}
		\label{eq:ineq2}
		\mu_p(x \in A | x \supset x') > \frac{1}{2}+\delta 
	\end{equation} 
for some $\delta = f(K)$.
	
	From Lemma~\ref{lem:n-edges} we realize that at least $n$ edges are required in the task set graph to make a task set unschedulable. Task sets that are unschedulable at higher utilization levels (higher than the unschedulable task set with minimum utilization) will have more edges in the task set graph. This observation helps us eliminate the possibility of an $x' \in A$ of constant size because the size of the minimal unschedulable task graph increases as we increase the number of tasks. In other words, inequality~\eqref{eq:ineq1} does not hold.

	Inequality~\eqref{eq:ineq2} cannot be true because that would imply that even assigning a very small utilization to certain tasks is bound to increase the probability of unschedulability by an additive constant. Let us assume that $K$ edges exist a priori in a task set graph. By Lemma~\ref{lem:n-edges}, we know that at least $n$ edges are needed for an unschedulable task graph, and each task (period) should receive at least one edge. The conditional probability that each task (period) gets at least one edge given that $K$ edges exist a priori and have been assigned in the best possible way (an edge each to $K$ periods) is still dependent on the total number of edges, $M$, which is greater than $n$. Thus the influence of a constant number of edges in the task set graph can not increase the probability of inducing unschedulability by a constant $\delta = f(K)$.
					
	When both inequalities~\eqref{eq:ineq1} and \eqref{eq:ineq2} do not hold, by a contrapositive argument, we cannot have $p\frac{d\mu_p(A)}{dp} < K/10$. (The other two prerequisites for Bourgain's theorem are definitely true.) Since $p\frac{d\mu_p(A)}{dp}$ is not bounded, we conclude that schedulability has a sharp threshold.
\end{proof}

The structure of the proof above is that, for task set graphs, premises (A1) and (A3) from Bourgain's theorem (Theorem~\ref{th:bourgain}) hold and conditions (C1) and (C2) are false therefore (A2) must be false and $p\frac{d\mu_p(A)}{dp}$ indicates a sharp threshold.

With edge probabilities being related to the utilization $U$, we can also use the term $\mu_{U}(A)$ to represent the probability that a task set with utilization $U$ is schedulable.

\begin{remark}[Width of the threshold interval] 
\label{rem:width}
As $n \rightarrow \infty$, the sharp threshold theorems indicate that the transition will be swift and going past the threshold will cause an immediate change in the ability to find the property of interest. For moderate values of $n$, it is possible to obtain some understanding of the swiftness of the transition. The width of the threshold interval is the smallest difference $U_1 - U_2$ such that $\mu_{U_1} (A) = \epsilon$ and $\mu_{U_2}(A) = 1 - \epsilon$ for a fixed $\epsilon, 0 < \epsilon < 1/2$. The width appears to be related to the number of permutations that are possible on the random structure. For the scheduling graph, the valid permutations correspond to permutations of the task set, i.e., among the $n$ tasks. Based on the work by Friedgut and Kalai (see Section 5 of their article~\cite{FK:96}), we conjecture that for a task set with $n$ tasks, the width of the threshold interval is $O(\frac{1}{\sqrt{n}})$.
\end{remark}

\begin{remark}[Location of the threshold]
Given a finite (but large) number of options for task periods, we have shown that there exists a sharp threshold for rate monotonic scheduling. The location of the threshold does depend on the number of tasks and the task periods. When the number of task periods are large, and not chosen pathologically, the location of the sharp threshold indicates good processor utilization.
\end{remark}

\section{Empirical results and discussion}
\label{sec:expts}

\subsection{Threshold behavior}

\begin{figure*}[htb]
\centering
    \subfigure[Task sets generated using the {\sc UUniSort} method~\cite{BB:05}]{
        \scalebox{0.60}{
            \includegraphics{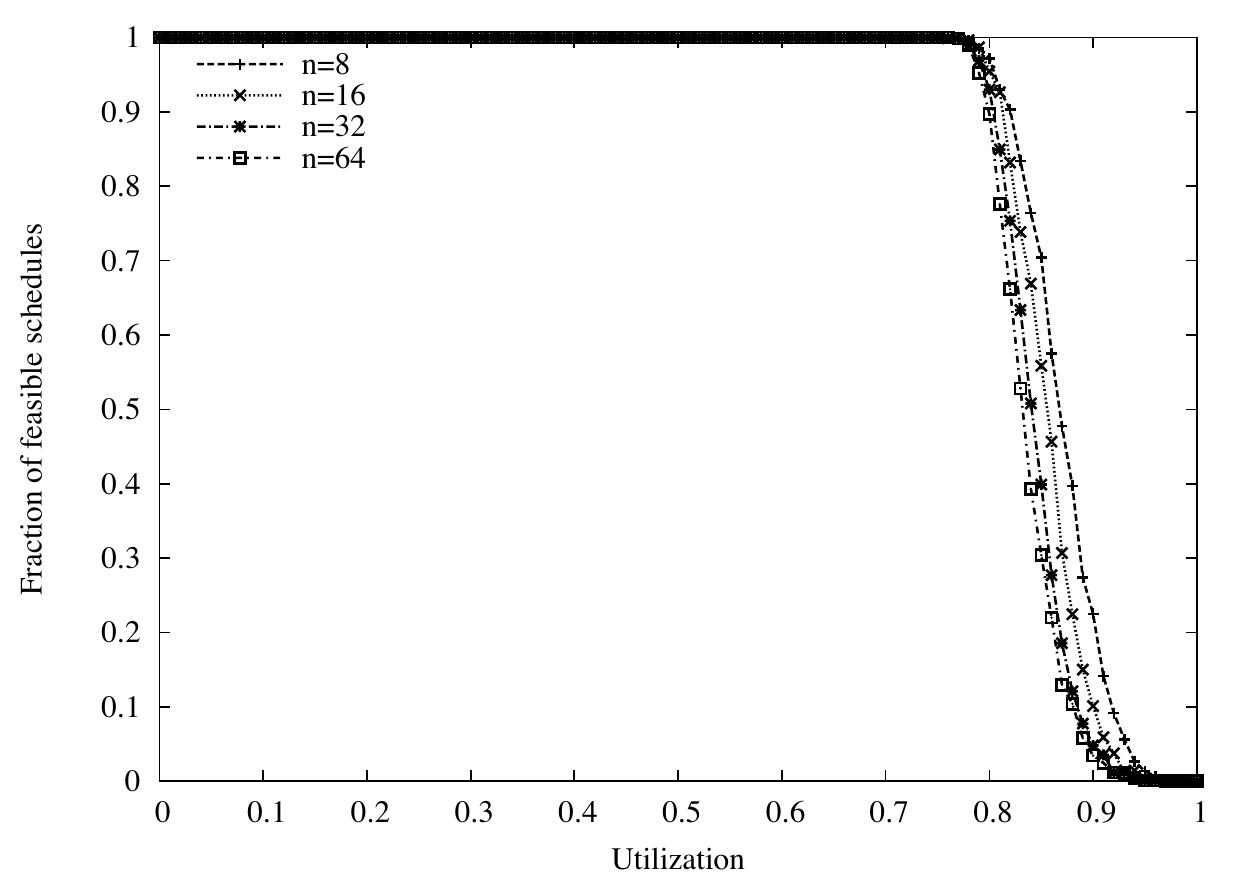}
        }
    }
    \subfigure[Task sets where each task had the same utilization]{
        \scalebox{0.60}{
            \includegraphics{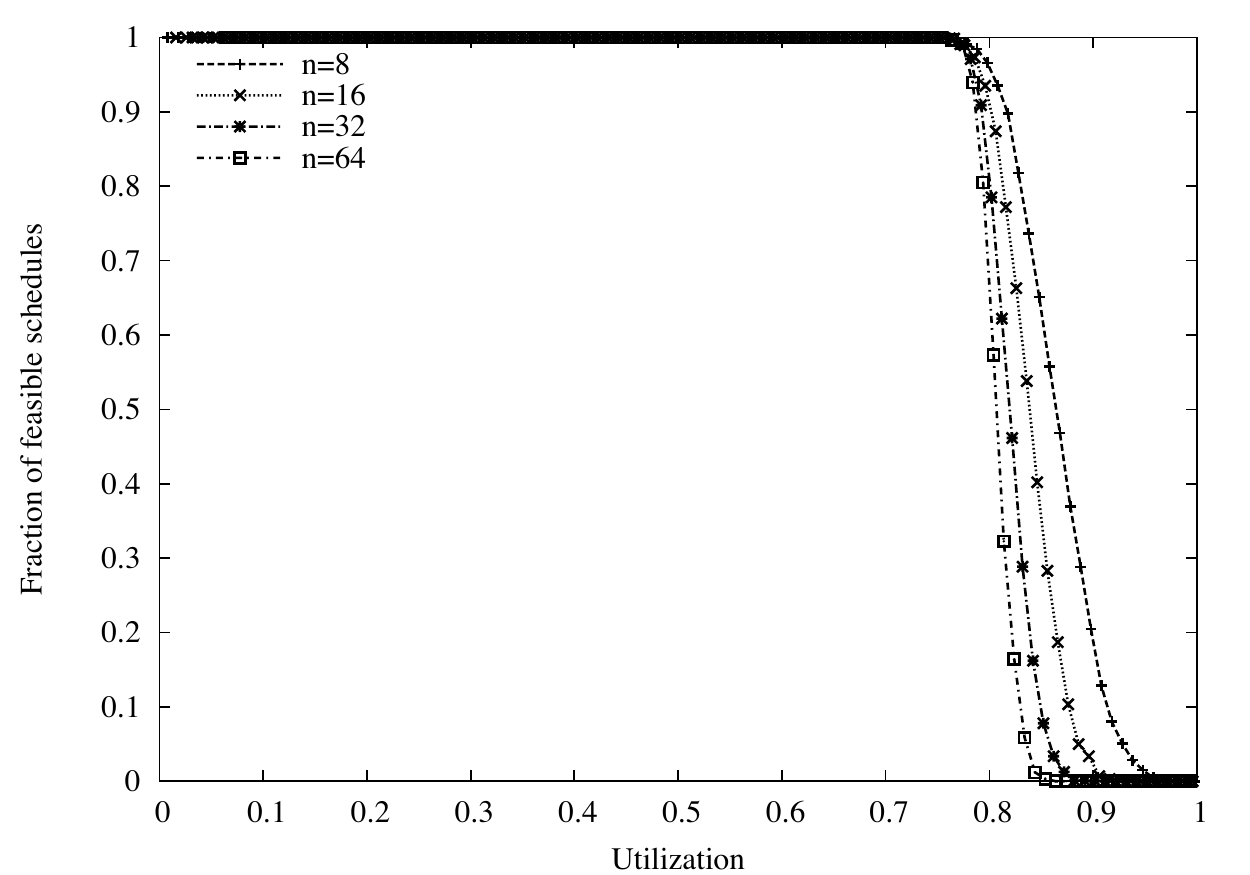}
        }
    }
  \caption{Thresholds for rate monotonic scheduling}
  \label{fig:graphs-uniprocessors}
\end{figure*}

Having established that rate monotonic scheduling has a sharp threshold, we use experiments to locate the threshold and to observe the swiftness of the transition from schedulability to unschedulability. A closed-form solution for the threshold has been elusive, and empirical evidence is our resort.

When examining experimental data, it behooves us to recall that sharp threshold behavior is a property of large task sets. For moderate size task sets, one can observe a threshold but it may not be as sharp as one would expect. (We present only a limited number of graphs for space considerations. Given the immense number of graphs that can be obtained, those shown here are intended as a visual cue to the theoretical machinery we have used. The results presented here  suggest that task sets of nominal sizes have a usable threshold.)

Bini and Buttazzo~\cite{BB:05} have studied different approaches to generating random task sets and have suggested methods with (almost) no bias. The goal of Bini and Buttazzo's work was to generate task sets uniformly at random from the space of all possible task sets that achieve utilization $U$. We employed the {\sc UUniSort} procedure from the article by Bini and Buttazzo~\cite{BB:05}. Periods were then drawn uniformly at random from $[1,10^{5}]$. Task set utilization was varied in steps of $0.1$ and at each level we tested $10^{4}$ task sets. The different numbers of tasks in a task set for the experiments were $8, 16, 32$ and $64$. Notice (in Figure~\ref{fig:graphs-uniprocessors}(a)) that schedulability drops rapidly when utilization is in the range $[0.8, 0.9]$. The width of the threshold interval is smaller for larger task sets. Within a rather short interval, we go from almost all task sets being schedulable to almost no task set being schedulable. This transition allows us to approximate the schedulbility test by using a utilization threshold near $0.8$.

We also conducted another experiment where we generated task sets that had the same utilization for each task: in other words, the total utilization was divided equally among all tasks. This experiment is informative because critically schedulable task sets for rate monotonic scheduling have this property~\cite{LL:73,BB:05}. The results of this experiment reveal (Figure~\ref{fig:graphs-uniprocessors}(b)) that when period values are arbitrary the achievable utilization is significantly higher than tight utilization bounds and that the threshold between schedulability and unschedulability is sharper.

The sharp utilization threshold result appears remarkable because it makes no assumptions about task periods and yet provides quite a precise estimate of schedulability. The general methodology for deriving utilization bounds for any scheduling policy involves identifying a task set that achieves low utilization and is yet unschedulable. It is not always easy to isolate the worst-case task set and determine its utilization. A major payoff from Theorem~\ref{th:main} is the ability to obtain thresholds empirically. When the worst case is rare (a low probability event) we are not burdened with a low utilization bound.

A possible concern is the asymptotic nature of the result. Sharp threshold behavior occurs when the number of tasks is large. We contend that this is exactly the case for which existing real-time
scheduling results are often inefficient (high complexity for analysis). As experiments reveal, a moderate number of tasks is sufficient for observing sharp thresholds. For small task sets, even
exact tests may be performed very quickly. There is a dependency between the threshold and the number of tasks. It is easily possible to compute -- offline -- the threshold for different numbers of tasks and utilize the appropriate threshold.

The use of thresholds becomes extremely useful in the case of soft real-time systems and for performing fast exploration of design space in developing (near-)optimal systems. An example is radar dwell scheduling~\cite{GRHL:04, GCSLS:04, GC:06b}. There are many task parameters that need to be tuned in a radar system to minimize tracking error subject to schedulability but the scheduling algorithms are hard to analyze; using thresholds for these problems simplifies the online optimization. Because performance is controlled at run-time, optimization routines cannot invoke exact tests that have high time complexity. Apart from online optimization, thresholds can be used as offline guidance measures to improve system designs.

\subsection{Some comparisons}

In our work, we make no assumptions about the task periods and execution times: they can be arbitrary. There has been work by Park, Natarajan and Kanevsky~\cite{PNK:96} and Lee, Sha and Peddi~\cite{LSP:04} obtained good utilization bounds by using task periods alone; execution times of tasks were unknowns in their approach. To determine if there is an improvement in coverage due to sharp threshold behavior, we assumed that tasks are restricted to periods in the set $\{3,8,11,16,20,42,120,300\}$; this set of periods has a utilization upper bound of $0.88$ using the technique of Lee et al.~\cite{LSP:04}. Generating tasks as we did earlier (using the {\sc UUniSort} approach), we found that the sharp threshold is about $0.94$, which is a $6\%$ improvement in utilization compared to the utilization upper bound obtained (Figure~\ref{fig:known-periods}). Techniques that use period information to obtain utilization bounds are effective but sharp threshold behavior allows us to be more aggressive even when period information is available. These results also indicate that sharp thresholds do exist even if periods are drawn from a restricted set.

\begin{figure}[htb]
\centering
	\scalebox{0.60}{
    	\includegraphics{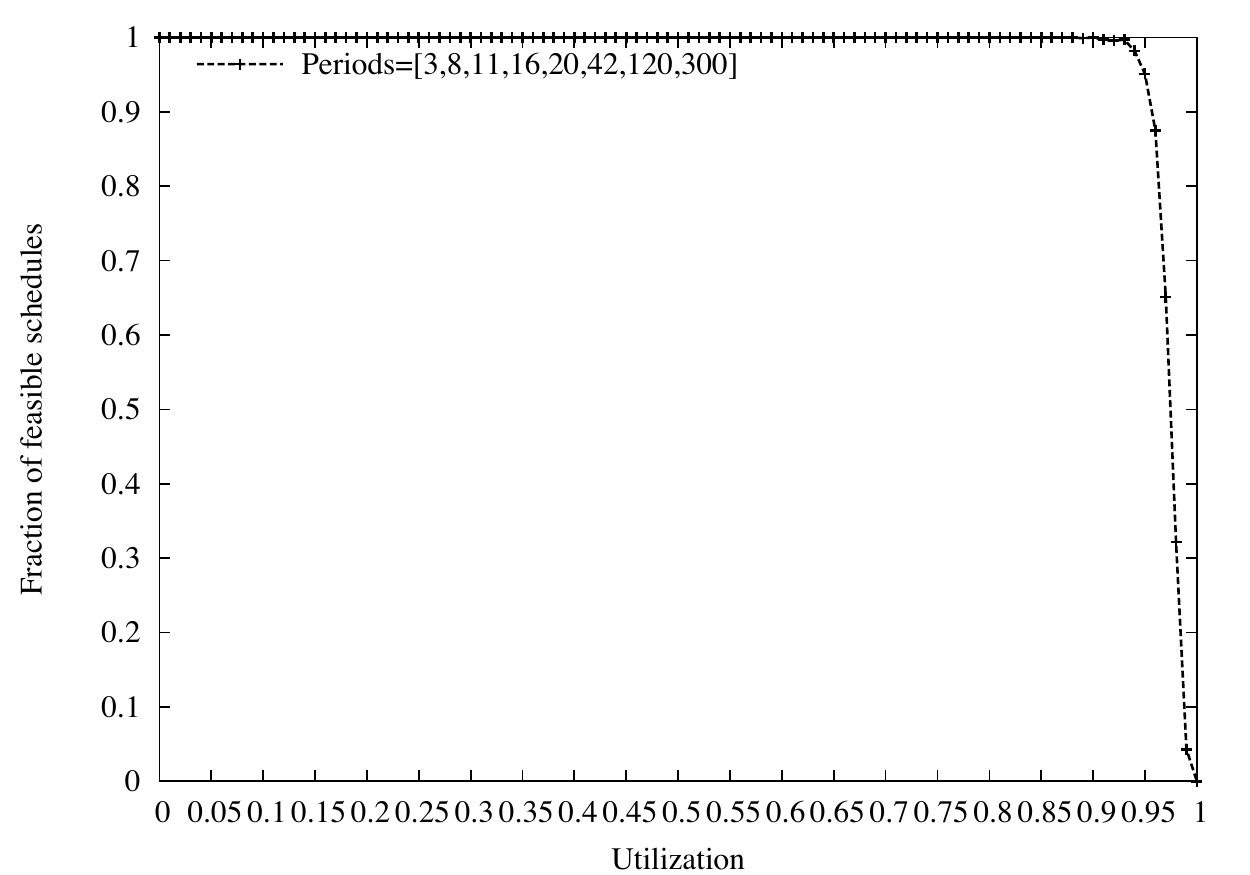}
    }
  \caption{Thresholds with known periods}
  \label{fig:known-periods}
\end{figure}

\section{Aperiodic workload and web server QoS}
\label{sec:aperiodic}

So far we have discussed rate monotonic scheduling of periodic tasks. In this section we extend the sharp threshold result to the aperiodic task model and highlight an application of this idea to the improved power management in delay-sensitive web services.

\subsection{Threshold for static-priority scheduling of aperiodic tasks}

We can look beyond periodic tasks and consider tasks that do not have a strictly periodic arrival pattern. Such a model has been investigated by Abdelzaher, Sharma and Lu~\cite{ASL:04} who derived {\em synthetic utilization bounds} for task sets where the execution times and relative deadlines for tasks are known. We can extend the theory of sharp thresholds to the case of aperiodic tasks easily. In this section we will establish that the schedulability of aperiodic tasks using the deadline monotonic priority policy has a sharp threshold and we will use this fact to improve on a power management scheme for web servers that was suggested by Sharma et al.~\cite{STASL:03}.

A job $i$ in an aperiodic task model has an arrival time $a_{i}$, an execution time $c_{i}$ and a relative deadline $D_{i}$ (the absolute deadline is $a_{i}+D_{i}$). Abdelzaher, Sharma and Lu define the synthetic utilization~\cite{ASL:04} of the set of active tasks at time $t$ as \[ U(t) = \sum_{\textrm{active tasks}} \frac{c_{i}}{D_{i}}\] where the set of active tasks at time $t$ is the set of tasks that were released at or before time instant $t$ and whose absolute deadlines are not earlier than $t$, i.e., $a_{i} \le t$ and $a_{i}+D_{i} \ge t$. If the synthetic utilization never exceeds a synthetic utilization bound, $U_{b}$, then all jobs are guaranteed to meet their deadlines~\cite{ASL:04,AL:01}. If $n$ is the maximum number of instances that can be active at any given time instant, we can show that there must exist a threshold $U^{*}$ such that task invocation patterns with $U(t) < (1-\epsilon)U^{*}$ are schedulable almost surely and task invocation patterns with $U(t) > (1+\epsilon)U^{*}$ are not schedulable almost surely as $n \rightarrow \infty$ for any $\epsilon > 0$.

It is useful to maintain a notion of job streams, which we will now define. 
\begin{definition}
An aperiodic job stream is a set of jobs where each job has the same execution time and relative deadline and job $j$ precedes job $k$ in the job stream iff $a_{j}+D_{j} \le a_{k}$.
\end{definition}
Essentially, only one instance of each job stream is active at a given time instant $t$.

\begin{theorem}
	The schedulability of aperiodic task streams, where each task stream $\tau_i$ is characterized by jobs with execution time $c_i$ and relative deadline $D_{i}$, has a sharp synthetic utilization threshold. The synthetic utilization at any time $t$ is $U(t) \le \sum_{i=1}^{n} \frac{c_i}{D_i}$.
	\label{th:main2}
\end{theorem}

The proof for the existence of a sharp threshold for deadline monotonic scheduling of aperiodic jobs does not require much deviation from the proof of the existence of a sharp threshold for rate monotonic scheduling of periodic tasks. The only modification that is required is to replace the vertex partition $\mathcal{T}$ with vertices that abstract most characteristics of an aperiodic job stream. Let each vertex in this partition represent a stream with relative deadline $D_{i}$ and a sequence of arrival times for that stream of jobs. In our analysis of periodic tasks, the period of a task was sufficient information to associate with each vertex. If we limit the arrival times to be integers in the interval $(0, T]$ for some integer $T$, we have a finite number of such vertices in $\mathcal{T}$. As $T \rightarrow \infty$, the number of vertices in $\mathcal{T}$, $n \rightarrow \infty$. $n$ is the number of job streams and hence is the maximum number of active jobs at any time instant. We can then use the same mechanism as before, with the task set graph, to show that a sharp threshold must exist for deadline monotonic scheduling because deadline monotonic scheduling satisfies the monotonicity property. To further confirm this knowledge, we generated many random instances of the aperiodic task scheduling problem and determined if jobs missed their deadlines. For these experiments, we had a varying number of job streams with the inter-arrival time for each job stream being drawn from an exponential distribution. The maximum synthetic utilization contribution of any one job stream (the maximum value of $c_{i}/D_{i}$) was kept at $0.125$ to allow for a sufficient number of streams. This is a modest assumption given that we would like to demonstrate the use of sharp thresholds to control the power consumption of a web server dealing with many (100s to 1000s) small jobs.

\begin{figure}[htb]
	\centering
	\scalebox{0.33}{
		\includegraphics{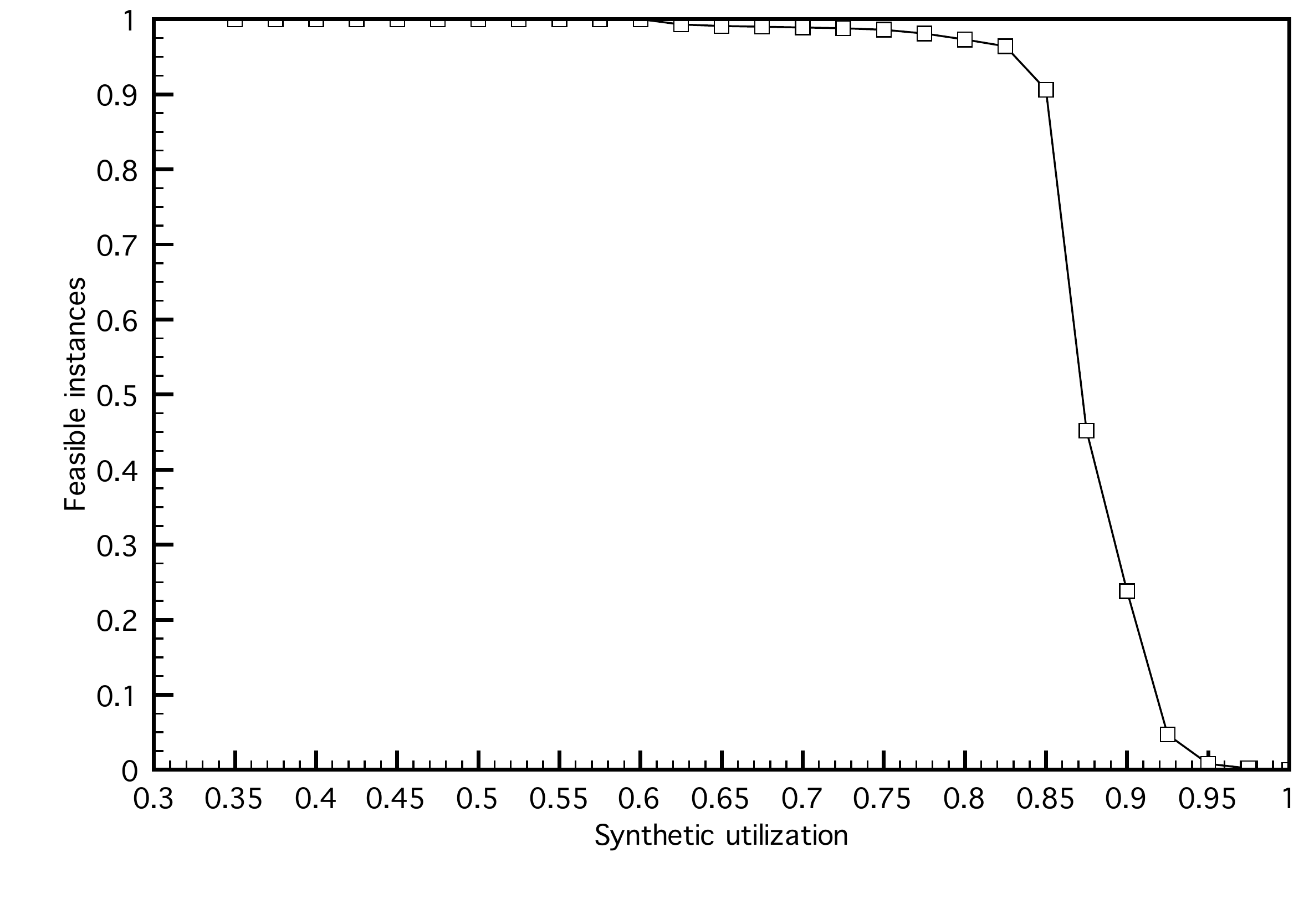}
	}
	\caption{Threshold for deadline monotonic scheduling of aperiodic tasks}
	\label{fig:aputil}
\end{figure}

The graph illustrating sharp threshold behavior for aperiodic task scheduling (Figure~\ref{fig:aputil}) indicates that the threshold for deadline monotonic scheduling may be close to $0.8$, which is substantially higher than a synthetic utilization bound of $0.586$ that can be obtained using worst-case analysis~\cite{AL:01,ASL:04}. By exploiting this difference between the average case and the worst case behavior of the deadline monotonic scheduling policy for aperiodic tasks, we can reduce power consumption for web servers without significant loss in temporal guarantees.

\subsection{Power control for web servers}

Many web services offer some delay bounds to clients as a part of the service level agreements; this is particularly true for services that require user fees. Moreover, web services offer multiple levels of service with better guarantees for premium customers. Synthetic utilization bounds are an effective mechanism to ensure that delay guarantees are met. Servers can use an admission control mechanism to ensure that they can limit the delay experienced by different clients. Alternatively, these bounds can be used to provision a web farm to ensure that all customer requirements can be met at low cost. 

Another application of such bounds is in operating power control. Most processors being manufactured today can operate at multiple clock speeds, with lower speeds consuming less power. Thus, utilization bounds can help in determining the ideal speed settings for processors such that delay bounds are not violated and power consumption is reduced. This approach was adopted by Sharma et al.~\cite{STASL:03}, and is illustrative of the use of synthetic utilization bounds. We will not stress the need for power control in server farms. The case has been made by many researchers including Sharma et al.~\cite{STASL:03}. {\em The only goal of this section of our article is to suggest that using synthetic utilization thresholds will improve power savings at the cost of a small fraction of deadline violations.} Sharma et al. used a synthetic utilization bound of $0.586$ for the web server, while we allow the web server to operate up to a synthetic utilization of $0.75$.

Tasks are scheduled using the deadline monotonic priority assignment, therefore different relative deadlines correspond to different service levels. We do not rewrite a web server like Apache to support multiple levels but, instead, run multiple instances of the Apache web server at different priority levels in the operating system\footnote{Most operating systems including Linux allow users to set static priorities for tasks. Within each priority level, tasks are scheduled FIFO by default.}, to provide service class differentiation. Our implementation is for the Linux operating system (Fedora Core 3; Linux kernel 2.6.9) and makes use of the TUX in-kernel web server~\cite{Tux} to integrate admission control, power control and scheduling.

\begin{figure}[htb]
	\centering
	\scalebox{0.45}{
		\includegraphics{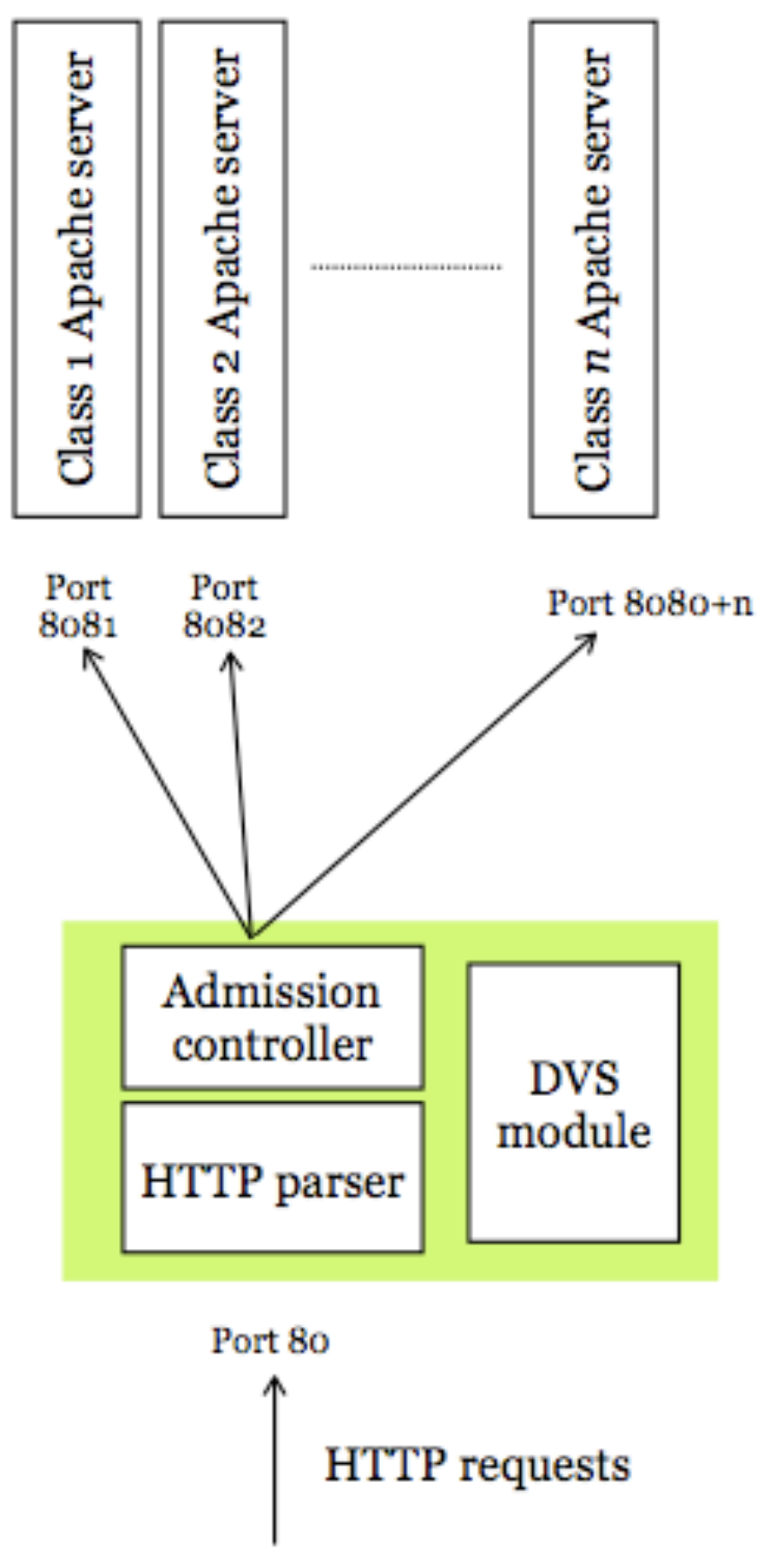}
	}
	\caption{A system architecture for web services}
	\label{fig:sysarch}
\end{figure}

All new HTTP session requests arrive and are processed by the TUX server. Based on the source of the request (or other meta information), a service level -- a delay guarantee, $D_{i}$ -- is assigned for the request. The service time, $c_{i}$, associated with a request is inferred from the content that is requested. If the new connection will not violate the synthetic utilization limit for the system (we chose $0.75$), the request is admitted. The service time for a request depends on the processor speed. If, at the current speed, the utilization limit is exceeded then the power control module uses dynamic voltage scaling to increase the processor speed and keep the synthetic utilization under the limit. When the processor is operating at the maximum speed, new HTTP connections may be rejected to keep the system operating under the set limit. Admitted sessions are handed off to the appropriate Apache server.

When a session terminates, it may be possible to reduce processor speed. We do not reduce the processor speed at once but wait for a predefined duration before making changes. This is to minimize overhead from rapid voltage changes. To keep track of the synthetic utilization after connections have been admitted, we make use of the Netfilter framework and some extra modules  that we implemented to track packets and identify HTTP traffic. The overall architecture of the web server platform (Figure~\ref{fig:sysarch}) is the same as the one used by Sharma et al.~\cite{STASL:03} and they have provided several implementation details that we do not discuss in this article but can be obtained from their report. The alterations we needed to make were only due to changes in the underlying platform.

We used an Intel Pentium M processor with enhanced SpeedStep capability and a maximum processing speed of 1.7 GHz. The \texttt{cpufreq} driver for enhanced SpeedStep allows us to control the operating speed. The TUX in-kernel web server is part of the Linux Fedora Core 3 distribution. In contrast, Sharma et al.~\cite{STASL:03} used an AMD Athlon processor with PowerNow DVS support. They also used Linux kernel 2.5, for which they needed to port \texttt{khttpd}, the in-kernel web server from Linux kernel 2.4.\footnote{There was a decision to remove the in-kernel web server between versions 2.4 and 2.5 of the Linux kernel, but the web server was brought back in to the 2.6 kernel by some distributions including Fedora.} The processor frequency and voltage settings for the processor we used are shown in Table~\ref{tab:dvs}.

\begin{table}[htb]
	\centering
	\begin{tabular}{ll}
		\hline 
		Frequency & Voltage \\
		MHz & Volts \\
		\hline
		600  & 0.956 \\
		800  & 1.004 \\
		1000 & 1.116 \\
		1200 & 1.228 \\
		1400 & 1.308 \\
		1700 & 1.484
	\end{tabular}
	\caption{Frequency and voltage settings/Intel Pentium M 1.7GHz with enhanced SpeedStep}
	\label{tab:dvs}
\end{table}

\begin{figure}[htb]
	\centering
	\scalebox{0.3}{
		\includegraphics{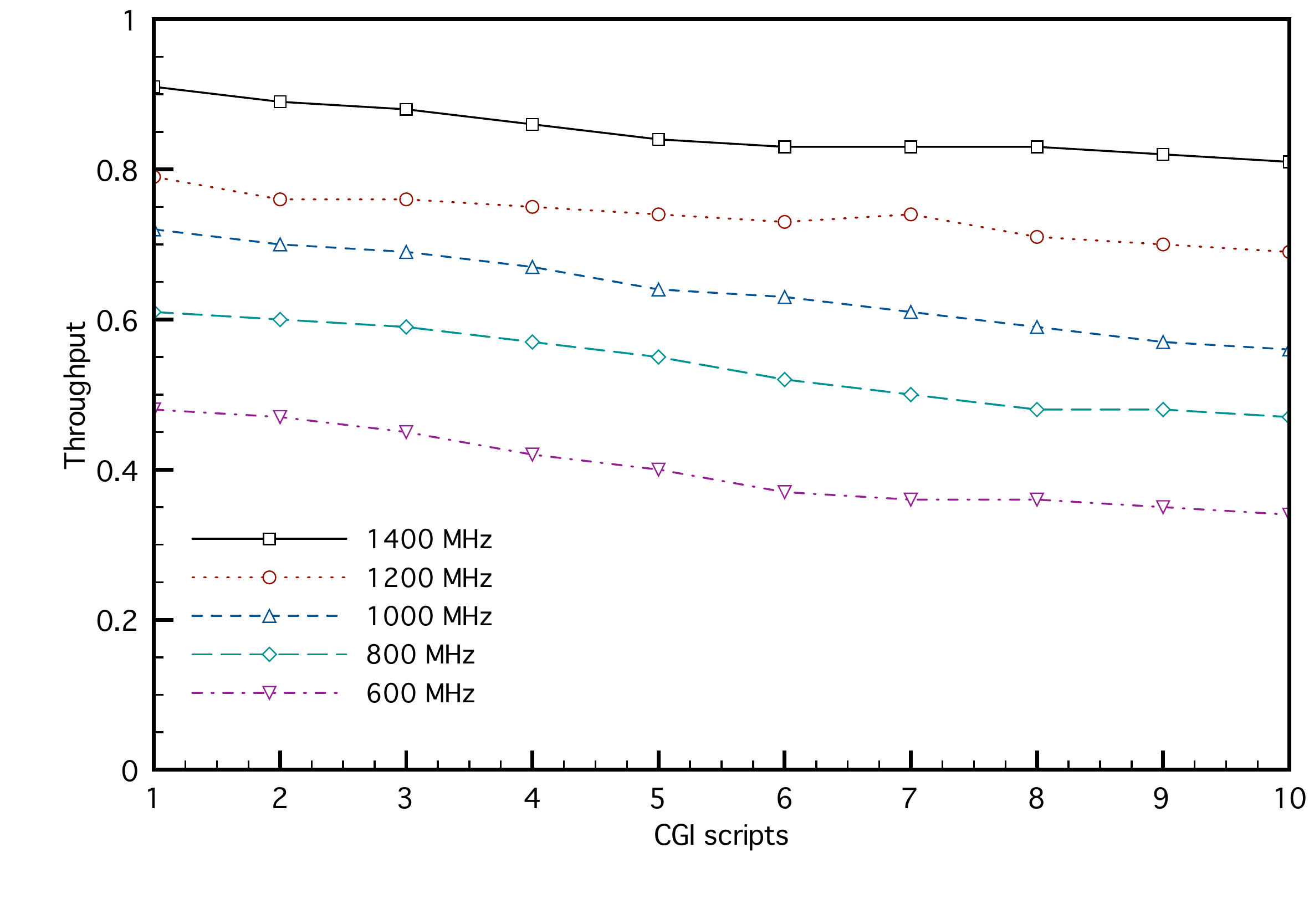}
	}
	\caption{Impact of processor speed on execution times}
	\label{fig:exectimes}
\end{figure}

The workload requested by different clients was composed of a set of CGI scripts that would be executed at the web server. We used 10 CGI scripts with varying degrees of computation. The execution time requirements of these scripts was determined by setting the processor speed at different levels and determining the maximum rate at which the processor could serve each CGI request. If, for example, at 1.7GHz, the server could handle 800 requests per second of script 1 alone, then the mean execution time of script 1 at this speed is $1/800 = 1.25ms$. We profiled the scripts at each of the six possible speed settings to determine the change in execution times with slowdown. This type of profiling helps us account for other execution sources of overhead including data I/O. The throughput slowdown (and hence the execution time increase) for each of the 10 scripts is illustrated (Figure~\ref{fig:exectimes}; the throughput at the highest speed is assumed to be $1$ and the throughput decrease at slower speeds is shown.)  The scripts to the right of the graph are computationally more demanding and we use the slowdown/speedup factors that correspond to these scripts when adjusting voltage levels.

To determine power savings, we used logs of session-oriented connections that were fed to {\tt httperf}~\cite{MJ:90} to generate workload for the web server from multiple clients. The workload we used contained 1000 persistent HTTP connections with random connection lengths chosen in the interval $[2,16]$. The requests could be for any of the profiled scripts and the inter-arrival time was drawn from an exponential distribution with different means. We created six Apache servers at different priority levels thus limiting the number of possible relative deadlines to six. 

There are two quantities of interest: the average power consumption and the fraction of deadlines missed. The average power consumption was obtained using a separate data acquisition system that measured the voltage drop across a sense resistor. For the same workload, we determined the average power consumption when the synthetic utilization set point was $0.586$ (the synthetic utilization bound) and $0.75$ (near the sharp threshold). It is clear that we can obtain power savings, and these savings are shown in Figure~\ref{fig:powersave}. The load (along the $x$-axis) is a fraction of the processor capacity based on the execution time profiling carried out earlier and the known inter-arrival times between HTTP requests. Increasing the load increases the maximum synthetic utilization. We varied the load from $0$ to $0.8$ and observed that we can save an additional 10-11\% energy by using a higher synthetic utilization set point. Using a set point of $0.586$, we noted slightly less than $1.7\%$ deadline misses and by raising the set point to $0.75$ we recorded $2.8\%$ deadline misses. Some deadline misses are inevitable, irrespective of the set point chosen unless we are overly conservative, because of variations in execution times and also depend on when exactly speed changes are performed. The encouraging result, however, is that we see greater power savings with a small penalty.

\begin{figure}[htb]
	\centering
	\scalebox{0.3}{
		\includegraphics{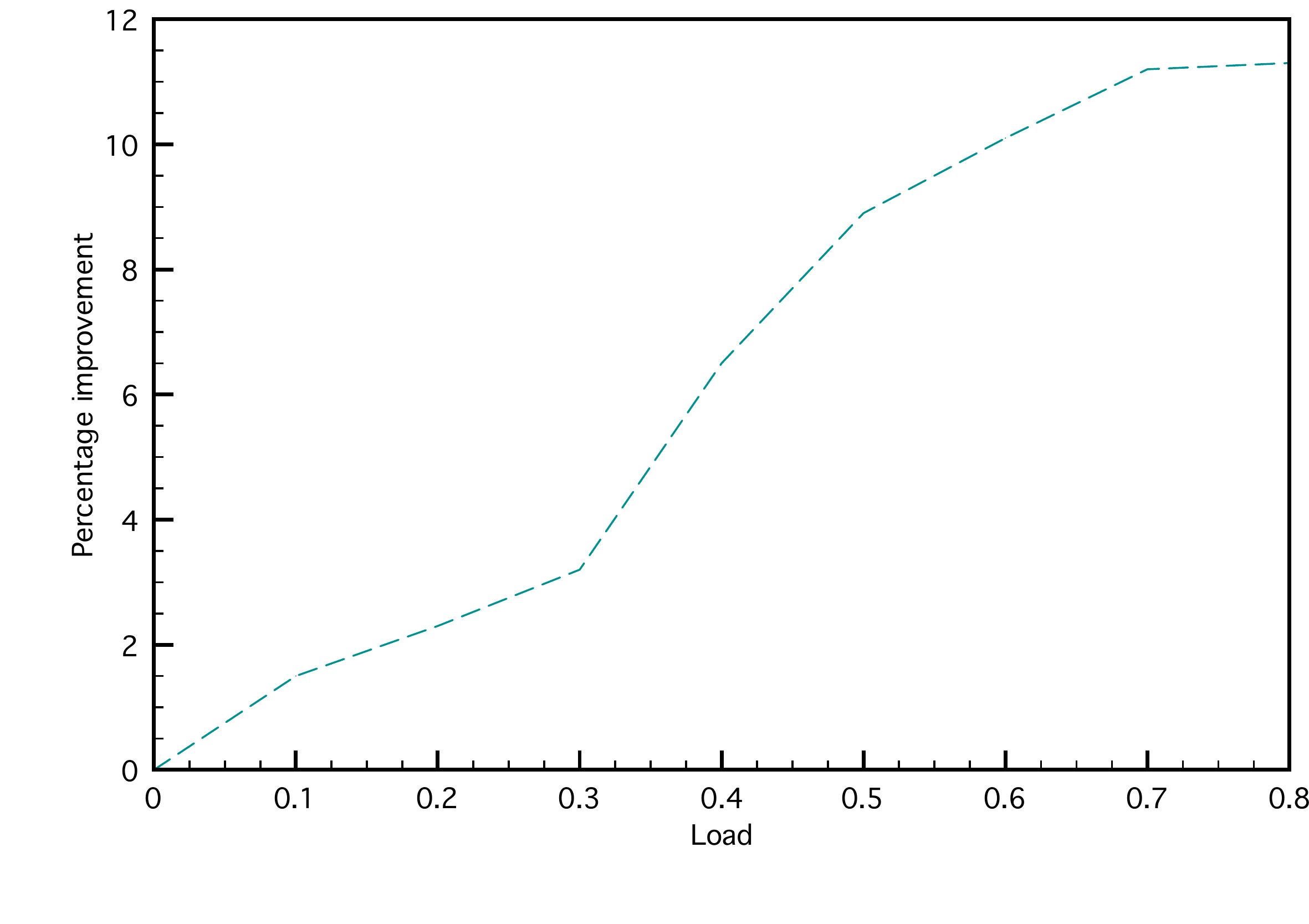}
	}
	\caption{Energy savings: sharp threshold vs. tight bound}
	\label{fig:powersave}
\end{figure}

By changing the synthetic utilization set-point from $0.586$ to $0.75$, a change of about $0.17$, we would expect to see power savings of that magnitude at moderate and high workload conditions. This does not happen because of the discrete frequency-voltage settings, which often forces the processor to operate at higher speeds. Yet another question involves the selection of the synthetic utilization set point. From earlier experiments (Figure~\ref{fig:aputil}), we could have picked a higher set point, say $0.80$. Even in the earlier experiments, the probability of missing a deadline is higher for a synthetic utilization of $0.80$, and if we do use this for the web server system the percentage of jobs missing their deadlines increased to $9\%$, which is significantly high.

The essential takeaway from this section is that the existence of sharp thresholds allows us to improve the management of computer systems. With web servers, we can either reduce the energy costs or (quite naturally) deal with additional workload with existing infrastructure.

\section{Related work}
\label{sec:relatedwork}

In our work, we explore some interesting aspects surrounding task set utilization and schedulability for real-time systems. There has been extensive work on deriving utilization bounds for periodic task systems starting with the work of Liu and Layland~\cite{LL:73}. Kuo and Mok~\cite{KM:91} made significant improvements on Liu and Layland's bound for rate monotonic scheduling by showing that schedulability is a function, not of the number of individual tasks but, of the number of harmonic chains. Bini, Buttazzo and Buttazzo~\cite{Bini:03} have shown, using the {hyperbolic bound}, that the feasible region for schedulability using the rate monotonic scheduling policy can be larger if the product of individual task utilizations (and not their sum) is bounded. Wu, Liu and Zhao used techniques inspired by network calculus to derive schedulability bounds~\cite{WLZ:05} for static priority scheduling. Their contribution is an alternative framework for deriving utilization bounds.

Our work presents a fresh perspective on scheduling for real-time systems. Only Lehoczky, Sha and Ding~\cite{LSD:89} have attempted to obtain average-case results. For rate monotonic scheduling, they characterized the breakdown utilization of the rate monotonic policy for the Liu and Layland model of real-time tasks as $0.88$. Breakdown utilization, however, is not the same as a utilization threshold, and the connection between the two needs to be examined more closely. The methodology we employ in obtaining our results is new and extremely general. It was not possible to reason in a rather abstract sense about the average-case behavior of scheduling policies with the more traditional analysis techniques of time demand and resource supply. Furthermore, our abstraction allows for reasoning about multi-stage and multiprocessor systems. Dutertre~\cite{Dutertre:02} identified phase transitions in a non-preemptive recurring task scheduling problem. While Dutertre's work emphasized the empirical evidence for sharp thresholds, we have provided the mathematical basis for the existence of sharp thresholds.

Lehoczky pioneered the use of real-time queueing theory to predict the behavior of real-time scheduling policies -- specifically the earliest deadline first policy -- under heavy traffic conditions with stochastic workload~\cite{Lehoczky:96,Lehoczky:97}. RTQT uses powerful tools to determine deadline miss percentages in end-to-end tasks executing on a resource pipeline. We may be able to use RTQT to predict the extent to which deadlines can be missed when a task set has utilization close to the threshold, but that requires extensive study, especially to extend RTQT to static priority policies.

In the realm of aperiodic task sets, great progress has been made recently, by Abdelzaher et al., with the identification of aperiodic schedulability bounds for static priority scheduling~\cite{ASL:04}. The initial result obtained by Abdelzaher and Lu~\cite{AL:01} was a constant time utilization-based test for a set of aperiodic tasks. The original analysis has been extended to deal with end-to-end schedulability for multi-stage resource pipelines~\cite{ATL:04}. It has also been shown that such analysis can be used to obtain non-utilization bounds for schedulability with static priority policies~\cite{LA:06}. In this article, we have studied single-node thresholds for the aperiodic task model. In the future, we will further the ideas described in this article to include resource pipelines and non-utilization metrics.

For the specific application of power control in web servers and web server clusters, there has been recent work by Bertini, Leite and Moss{\'e}~\cite{BLM:07}, and Horvath, Abdelzaher and Skadron~\cite{HAK:07}; we believe that the ideas proposed here can easily be integrated into these resource management solutions.

Sharp thresholds are indicators of phase transitions. Phase transitions are common in physical systems. Freezing of ice and superconductivity are phenomena that have temperature as the critical parameter. Phase transitions have been identified in many combinatorial optimization problems, especially constraint satisfaction problems~\cite{CKT:91,MSL:92,KS:94}. Phase transitions provide very interesting insight into the behavior of combinatorial optimization problems, of which scheduling is an instance, and mayhold the key to faster, near-optimal solutions. Sharp thresholds for properties of random graphs were identified initially by Erd{\"o}s and R{\'e}nyi~\cite{ER:60} and these results have been generalized by many mathematicians including Friedgut and Kalai~\cite{FK:96, Friedgut:99}.

\section{Conclusions}
\label{sec:conclusions}

The search for efficient tests for schedulability has been at the center of real-time systems research. We have generalized the use of utilization as a schedulability metric. By identifying the sharp threshold behavior of scheduling policies with respect to utilization, we provide a new test for schedulability. Schedulability tests using utilization thresholds are well-suited for soft real-time systems. For hard real-time systems these tests can be backed up by exact tests; thresholds can be used to perform initial filtering before using exact tests.

Most scheduling policies can be shown to have sharp thresholds. We have introduced the task set graph abstraction that can be used to argue about the average case behavior of policies irrespective of whether the workload is periodic or aperiodic. This abstraction is powerful enough to reason about uniprocessor scheduling, and we expect to apply the same ideas to multiprocessor and multistage scheduling problems, and a variety of policies although we considered only the rate and deadline monotonic priority policies in this paper. Interestingly, we have been able to use these thresholds to improve the energy efficiency of delay-sensitive web servers.

Our approach to dealing with average or typical case behavior of scheduling policies makes interesting connections with results from percolation theory and random graphs. We hope to explore these links further to fully characterize the performance of scheduling policies. So far, we have been able to make some qualitative statements about scheduling policies but the ability to compare policies, which we have not explored with this framework, will enrich the graph-theoretic approach.

There are several related open problems. The first of these is the determination of the threshold for a policy without having to resort to experiments. Related to this is the secondary issue of determining the width of the threshold interval. The analysis is complex because of the time demand function that is needed to evaluate the completion time of a task. In a strictly periodic setting with rate monotonic scheduling, the completion time of a task $\tau_{i}$, $L$, is obtaining by fixed point iteration. \[ L = \sum_{i=1}^{i} \lceil \frac{L}{P_{i}} \rceil c_{i}, \] where the summation is taken over all tasks with priorities greater than or equal to the task $\tau_{i}$. $L \le P_{i}$ is necessary and sufficient for $\tau_{i}$ to meet its deadline. We believe that developing some normal approximations will provide a better understanding of the threshold for the rate monotonic policy, as well as other policies. Another useful result would be a measure of the worst-case tardiness over all possible task sets when the utilization is known. When the utilization is less than the Liu and Layland bound~\cite{LL:73}, the tardiness is always zero but little is known about the worst possible tardiness for arbitrary utilization factors.


\bibliographystyle{acm}
\bibliography{../../Bibliography/CompleteBibliography}


\end{document}